%% file: main.tex
\documentclass[conference]{IEEEtran}
\IEEEoverridecommandlockouts
\usepackage{cite}
\usepackage{amsmath,amssymb,amsfonts}
\usepackage{graphicx}
\usepackage{xcolor}
\usepackage{booktabs}
\usepackage{balance}
\usepackage{microtype}
\newcommand{\E}{\mathbb{E}}

\newcommand{\LOO}{\mathrm{LOO}}
\newcommand{\Sh}{\mathrm{Sh}}

\newtheorem{proposition}{Proposition}


\begin{document}

\title{Receiver-Specific Interventional Valuation of Semantic Atoms for\\
Budget-Constrained Multi-Agent Communication}

\author{
\IEEEauthorblockN{Danial Dehghani, Maximilian L\"ubke, Friedemann Laue, Zahra Seifaei and Norman Franchi}
\IEEEauthorblockA{Institute for Smart Electronics and Systems (LITES)\\
Friedrich-Alexander-Universit\"at Erlangen-N\"urnberg\\
Email: \{danial.dehghani, maximilian.luebke, friedemann.laue, zahra.seifaei, norman.franchi\}@fau.de}
}

\maketitle

\begin{abstract}

Cooperative multi-agent semantic communication requires deciding not only
how a transmitter encodes its observation, but also which resulting semantic
content should be delivered to each receiver under a limited link budget.
We introduce an encoder-agnostic transmitter-side framework for
receiver-specific valuation of independently selectable semantic units,
termed semantic atoms. Each atom is assigned a signed value that quantifies
whether delivering it is expected to improve or degrade the intended
receiver's downstream decision. We formulate this effect through a
controlled-delivery utility and identify it via prospective randomized
payload delivery before the receiver acts. A shadow receiver estimates this
utility using only transmitter-available information, from which
receiver-specific leave-one-out (LOO) and Shapley atom values are computed.
Experiments on a four-receiver intersection benchmark show 68.4\%
accuracy in recovering the sign of non-negligible LOO atom effects and
a Spearman correlation of 0.302 between estimated and reference LOO
values. Under limited budgets, we show that receiver-specific selection
consistently outperforms shared and random alternatives, demonstrating the
benefit of adapting semantic ranking to the intended receiver. At full
affordability, exhaustive subset selection also outperforms full-payload
transmission, showing that transmitting all available content need not
maximize receiver utility. Overall, our results show that receiver-specific
semantic valuation can support effective budget-constrained communication
despite imperfect atom-value estimation.

\end{abstract}

\begin{IEEEkeywords}
semantic communication, generative semantic communication, value of information, multi-agent systems, task-oriented communication, budgeted selection
\end{IEEEkeywords}

\input{sections/01_introduction}
\input{sections/02_system_model}
\input{sections/03_method}
\input{sections/04_experiments}
\input{sections/05_results}
\input{sections/06_conclusion}
\input{sections/07_acknowledgment}

\balance
\bibliographystyle{IEEEtran}
\bibliography{references}
\end{document}

%% file: sections/01_introduction.tex
\section{Introduction}\label{sec:introduction}

Future wireless systems will increasingly connect vehicles, robots, sensors,
and autonomous agents that exchange information to support decisions rather
than merely reproduce source samples. Prior work on semantic and task-oriented
communication has established this shift toward transmitting information
according to its relevance to the downstream task
\cite{xie2021deepsc,kountouris2021}. More recently, generative semantic communication has extended this direction
by using foundation models, diffusion models, and multimodal priors for
more efficient semantic representation and reconstruction
\cite{liang2025genai,zhang2025successive}. While these advances improve semantic representation and recovery, an additional challenge arises in cooperative multi-agent systems: when an agent shares its observation with multiple receiving agents, it must determine which extracted semantic information to communicate to each receiver under a constrained link budget.

Consider a cooperative intersection in which a roadside agent and several connected vehicles form a multi-agent team. At a given epoch, the roadside agent may observe a pedestrian trajectory, an occluded cyclist, and a congestion cue, while the other agents differ in their routes, local observations, and decision states. Consequently, different receiving agents may require different parts of the same shared observation. Communicating the full representation to every agent can consume the limited per-link communication budget on redundant information, whereas communicating one common subset cannot account for these heterogeneous information needs. Thus, even when agents share the same source of semantic information, the communication relevance of its contents can differ across the team.

Existing semantic communication research has advanced learned semantic
representation \cite{xie2021deepsc}, task-oriented transmission based on
feature importance \cite{wang2024importance}, and adaptive communication
resource allocation \cite{zhang2023drl}. Generative semantic communication has further improved semantic reconstruction
from compact representations through learned priors
\cite{liang2025genai,zhang2025successive}, while task-oriented multi-agent
communication has used value of information to adapt the precision of shared
observations \cite{mostaani2025scale}. These approaches guide task-relevant encoding or resource assignment, but they do
not directly address assigning independently selectable components of one
transmitting agent's observation differently to multiple receiving agents
according to their individual information needs.

A closely related line of work studies learned communication and credit assignment in cooperative multi-agent systems. Targeted messaging and social-influence methods learn whom to communicate with or measure how communication affects other agents' behavior \cite{das2019tarmac,jaques2019social}, while counterfactual and Shapley-based approaches assign value to communicated messages \cite{gao2025cmv,xue2022shapleymv}. These works provide important foundations for communication valuation, but primarily operate at the message or link level rather than valuing independently selectable semantic content for different receivers under explicit communication budgets. We address this gap with a receiver-specific semantic valuation framework for
budget-constrained multi-agent communication. Rather than using one importance
measure across receivers, we evaluate communicated content with respect to the
intended receiver and its downstream decision. The framework links semantic
valuation directly to receiver-specific payload construction while remaining
independent of the underlying encoder, making it applicable to the outputs of generative,
multimodal, or conventional semantic representation models. It therefore enables receiver-specific selection of independently selectable
semantic content under explicit per-link budgets, so that different receivers
can receive different information from the same observation.

The remainder of this paper is organized as follows.
Section~\ref{sec:model} presents the system model,
Section~\ref{sec:method} develops the valuation and selection method,
Section~\ref{sec:experiments} describes the experiments,
Section~\ref{sec:results} reports the results, and
Section~\ref{sec:conclusion} concludes the paper.

%% file: sections/02_system_model.tex
\section{System Model}\label{sec:model}

We consider cooperative multi-agent decision making with a single pre-action communication round. This section first specifies the team dynamics and communication interface, then defines semantic atoms and receiver-specific payloads together with their receiver-specific task utility, and finally formulates communication cost, link-budget constraints, and the resulting budgeted payload-selection problem.

\subsection{Cooperative Multi-Agent Communication Setting}

We model a cooperative team as a decentralized partially observable Markov decision process (Dec-POMDP) \cite{oliehoek2016},
$\mathcal{M}=(\mathcal{N},\mathcal{S},\{\mathcal{O}_i\},\{\mathcal{A}_i\},T,\Omega,r,\gamma,\mu_0)$.
Here, $\mathcal{N}$ is the set of agents, $\mathcal{S}$ denotes the global state space, and $\mathcal{O}_i$ and $\mathcal{A}_i$ are the observation and action spaces of agent $i$; $T$, $\Omega$, $\gamma$, and $\mu_0$ denote the transition kernel, observation kernel, discount factor, and initial-state law, respectively.

At epoch $t$, the system is in state $s_t$, and a transmitting agent $i$ observes
$o_{i,t}\in\mathcal{O}_i$. The scheduler specifies the receiver set
$\mathcal{R}_{i,t}\subseteq\mathcal{N}\setminus\{i\}$ and the corresponding
per-link communication budgets. Before the current communication round, each
$j\in\mathcal{R}_{i,t}$ has a local history $\tau^{-}_{j,t}$ comprising the
information available to that receiver up to that point. A single simultaneous,
error-free communication round then allows agent $i$ to transmit to the scheduled
receivers. Scheduling, relaying, and shared-budget optimization are outside our
scope.

The information delivered in the current round is incorporated into the
receiver's pre-communication history $\tau^-_{j,t}$, yielding the
post-communication history $\tau^+_{j,t}$. Receiver $j$ then draws its action
$a_{j,t}$ from the stochastic policy $\pi_j$ conditioned on this updated
history, i.e.,
$a_{j,t}\sim\pi_j(\cdot\mid\tau^+_{j,t})$.
After communication, all agents take their task actions, forming the joint action
$\mathbf{a}_t=(a_{1,t},\ldots,a_{N,t})\in\mathcal{A}$, where
$\mathcal{A}\triangleq\prod_{i\in\mathcal{N}}\mathcal{A}_i$.
The agents constitute a cooperative team in that they pursue the same task objective and receive the shared reward
$r:\mathcal{S}\times\mathcal{A}\rightarrow\mathbb{R}$,
so that $r(s_t,\mathbf{a}_t)$ evaluates the team's outcome at epoch $t$.
The policy profile $\boldsymbol{\pi}=(\pi_1,\ldots,\pi_N)$ is pretrained and remains fixed during logging, valuation, and deployment; hence, only the communicated payload is optimized in this study.
\subsection{Semantic Atoms, Payloads, and Receiver-Specific Utility}
\label{subsec:atoms}

To enable communication decisions at a finer granularity than an entire observation, we represent the information available at transmitting agent $i$ as a collection of independently deliverable semantic units. We call each such unit a \emph{semantic atom}: an operationally minimal encoded component that is not further decomposed by the subsequent selection layer and can be delivered to, or withheld from, each receiving agent independently. The usefulness of a selected subset of atoms for a particular receiver is defined later in this subsection at the payload level.

We model semantic-atom extraction through a fixed upstream mapping
$\Phi_{\phi_i}$, whose design and training are outside the scope of this work.
Applied to observation $o_{i,t}$, the mapping generates the indexed candidate set
\begin{equation}
Z_{i,t}
\triangleq
\big\{(m,q_{i,m,t})\big\}_{m=1}^{M_{i,t}}
=
\Phi_{\phi_i}(o_{i,t}),
\label{eq:candidate-set}
\end{equation}
where $M_{i,t}=|Z_{i,t}|$ is the number of available atoms, $q_{i,m,t}$ is the encoded and quantized semantic descriptor of atom $m$, and index $m$ preserves atom identity. Thus, $Z_{i,t}$ denotes the \emph{full candidate set} of semantic atoms extracted from agent $i$'s current observation. It is formed before receiver-specific valuation and is therefore common to all $j\in\mathcal{R}_{i,t}$.

For each receiving agent $j$, the transmitted payload is a receiver-specific subset
$S_{i\to j,t}\subseteq Z_{i,t}$ of the common candidate set. Its composition is specified by the binary selection vector
$\mathbf{x}_{i\to j,t}\in\{0,1\}^{M_{i,t}}$, where
$x_{i\to j,m,t}=1$ indicates that atom $m$ is included. Accordingly,
\begin{equation}
S_{i\to j,t}
\triangleq
\big\{(m,q_{i,m,t})\in Z_{i,t}:
x_{i\to j,m,t}=1\big\},
\label{eq:payload}
\end{equation}
so that the same candidate set can induce different payloads for different receivers.

Having defined the receiver-specific payload $S_{i\to j,t}$, we next
characterize its usefulness to receiver $j$. This usefulness depends on the
receiver's pre-communication information $\tau^-_{j,t}$, summarized by a
private decision state $\beta_{j,t}$ that is not observable by transmitter
$i$. The transmitter instead observes a public context $c_{j,t}$, such as the
receiver's identity, role, route, task query, or public link metadata.

To quantify the contribution of the receiver's payload-informed action to the
cooperative task, we use the bounded difference reward
\cite{wolpert2001}
\begin{equation}
Y_{j,t}
\triangleq
r(s_t,\mathbf{a}_t)
-
\mathbb{E}_{\tilde a_j\sim\rho_{j,t}}
\!\left[
r\!\left(s_t,(\mathbf{a}_{-j,t},\tilde a_j)\right)
\right],
\label{eq:receiver-outcome}
\end{equation}
where $\mathbf{a}_{-j,t}$ contains the realized actions of all agents except
$j$, $\tilde a_j$ is a counterfactual action of receiver $j$, and
$\rho_{j,t}=\pi_j(\cdot\mid\tau^{-}_{j,t})$ is the same fixed receiver policy
evaluated without the current-round payload. Thus, $Y_{j,t}$ measures the
incremental task contribution of receiver $j$'s post-communication action
relative to its no-current-payload behavior, with the other agents' actions
held fixed.

For a fixed link $i\to j$, epoch $t$, candidate set
$Z\equiv Z_{i,t}$, and public context $c\equiv c_{j,t}$, we define the
usefulness of any payload $S\subseteq Z$ by the controlled-delivery utility
\begin{equation}
U_{j,t}(S\mid c,Z)
\triangleq
\mathbb{E}_{\beta_{j,t},\,\omega_t}
\!\left[
Y_{j,t}
\mid
\operatorname{do}(S_{i\to j,t}=S),\,c,\,Z
\right].
\label{eq:payload-utility}
\end{equation}
Here, the expectation averages $Y_{j,t}$ over
$\beta_{j,t}$ and the remaining randomness $\omega_t$. Since
$\beta_{j,t}$ is not observed by transmitter $i$, the utility is conditioned
only on the transmitter-available context $c$ and candidate set $Z$.
The variable $\omega_t$ collects the remaining policy, environment, nonfocal
communication, and action randomness. The
$\operatorname{do}$ operator denotes an external intervention that sets a
variable, rather than observational conditioning on its realized value
\cite{pearl2009overview}. Here,
$\operatorname{do}(S_{i\to j,t}=S)$ intervenes only on the focal link,
forcing $S$ before receiver $j$ acts, while all nonfocal links and system
dynamics retain their joint conditional distribution, including dependencies
among other agents' payloads. Averaging over $\beta_{j,t}$ and $\omega_t$
therefore isolates the expected effect of focal delivery without fixing
unrelated system realizations. Hence, $U_{j,t}(S\mid c,Z)$ is the expected receiver-attributed task
contribution under forced delivery of $S$ for the specified public context
and candidate set.
\subsection{Communication Cost and Budgeted Payload Selection}
\label{subsec:selection-problem}

Having defined receiver-specific payload utility, we now specify the
communication resources that constrain payload construction. Each atom is
independently packetized and assigned a known positive cost
$b_{i,m,t}>0$ bits, including encoded content and required metadata.
The cost is fixed before receiver-specific selection and known to
transmitting agent $i$.

For any $S\subseteq Z_{i,t}$, independent packetization and the exclusion
of cross-atom coding yield
\begin{equation}
C(S)\triangleq
\sum_{(m,q_{i,m,t})\in S} b_{i,m,t},
\qquad C(\emptyset)=0.
\label{eq:payload-cost}
\end{equation}
The scheduler assigns each active link $i\to j$ a budget
$B_{i\to j,t}\geq0$ bits for the current round; feasibility requires
$C(S_{i\to j,t})\leq B_{i\to j,t}$.

Utility and budget together define the transmitter-side selection problem that
our estimation and selection method aims to approximate. For each active link $i\to j$,
\begin{subequations}\label{prob:selection}
\begin{align}
\max_{\mathbf{x}_{i\to j,t}\in\{0,1\}^{M_{i,t}}}\quad&
U_{j,t}\!\left(S_{i\to j,t}\mid c_{j,t},Z_{i,t}\right)
\\
\text{s.t.}\quad&
\sum_{m=1}^{M_{i,t}}
b_{i,m,t}x_{i\to j,m,t}
\leq B_{i\to j,t}.
\end{align}
\end{subequations}
The zero vector selects $\emptyset$ and guarantees feasibility.
Problem~\eqref{prob:selection} is myopic and solved independently per
active link; scheduling, joint-link optimization, and shared budgets are
outside our scope.

Because atoms may be redundant, complementary, or harmful under the fixed
receiver policy, $U_{j,t}$ need not be additive or monotone. Thus, solving \eqref{prob:selection} generally requires considering the
utility of different atom subsets rather than simply ranking atoms by
individual importance. If $B_{i\to j,t}<C(Z_{i,t})$, subset selection is mandatory.
Even when $Z_{i,t}$ is affordable, transmitting all atoms need not be optimal.
Section~III estimates this utility and constructs score-based payloads.
Under modular utility, selection reduces to an additive knapsack; outside
modularity, score-based selection is evaluated empirically.
Proposition~1 later formalizes the full-payload cases after signed atom values
are defined.

%% file: sections/03_method.tex
\section{Semantic Atom Valuation and Selection}
\label{sec:method}

This section develops transmitter-side receiver-specific payload selection by
defining signed atom values, identifying and estimating them from randomized
delivery, and using the resulting scores for budget-feasible selection.
\subsection{Signed Interventional Atom Values}
\label{sec:estimand}

Section~\ref{sec:model} defines $U_{j,t}$ as the receiver-specific utility of
a complete payload. For payload construction, however, we must also determine
how much each individual atom contributes to that utility. We therefore derive
signed atom-level values from $U_{j,t}$ and use them as the quantities to be
estimated and selected under the communication budget.

For link $i\to j$ at epoch $t$, let $Z\equiv Z_{i,t}$, $c\equiv c_{j,t}$, and $M\equiv|Z|$, with $q_m\equiv(m,q_{i,m,t})\in Z$ denoting atom $m$. For any $S\subseteq Z\setminus\{q_m\}$, the signed marginal effect of atom $m$ is
\begin{equation}
\Delta_{i\to j,m,t}(S)
\triangleq
U_{j,t}(S\cup\{q_m\}\mid c,Z)
-
U_{j,t}(S\mid c,Z).
\label{eq:marginal}
\end{equation}
Thus, $\Delta_{i\to j,m,t}(S)$ quantifies the change in the controlled-delivery utility of receiver $j$ caused by adding $q_m$ to the reference payload $S$. Because this contribution can vary with the other atoms already contained in $S$, an atom does not generally have a unique context-independent marginal effect. We therefore obtain a single atom-level value by specifying how the reference payload is chosen.

Let $Z^{(-m)}\triangleq Z\setminus\{q_m\}$. We adopt the
leave-one-out (LOO) and Shapley attributions \cite{shapley1953}:
\begin{align}
\varphi^{\LOO}_{i\to j,m,t}
&\triangleq
\Delta_{i\to j,m,t}(Z^{(-m)}),
\label{eq:loo}\\
\varphi^{\Sh}_{i\to j,m,t}
&\triangleq
\sum_{S\subseteq Z^{(-m)}}
\frac{|S|!(M-1-|S|)!}{M!}
\Delta_{i\to j,m,t}(S).
\label{eq:shapley}
\end{align}
LOO measures the atom's marginal effect relative to the full payload,
whereas Shapley averages its marginal effect over all coalition
backgrounds and satisfies
$\sum_{m=1}^{M}\varphi^{\Sh}_{i\to j,m,t}
=
U_{j,t}(Z\mid c,Z)-U_{j,t}(\emptyset\mid c,Z)$.
This identity only decomposes the total utility gain; it does not make
Shapley-based additive selection optimal when atom interactions make
$U_{j,t}$ nonadditive.

The atom values are deliberately signed. Classical value of information
is nonnegative when additional information can be ignored by a
re-optimized decision rule \cite{howard1966}; here, however, the deployed
receiver policy $\pi_j$ remains fixed and may react adversely to stale,
conflicting, or otherwise misleading information. Consequently,
$\varphi^{\LOO}_{i\to j,m,t}<0$ means that withholding $q_m$ from the
full payload increases receiver $j$'s expected utility.

\begin{proposition}
\label{prop:fullpayload}
Fix $(i,j,t,c,Z)$. Suppose first that
$U_{j,t}(\cdot\mid c,Z)$ is monotone nondecreasing with respect to set
inclusion. If $B_{i\to j,t}\geq C(Z)$, then the full payload $Z$ solves
\eqref{prob:selection}. If $B_{i\to j,t}<C(Z)$, then $Z$ is infeasible.
Finally, if $B_{i\to j,t}\geq C(Z)$ and
$\varphi^{\LOO}_{i\to j,m,t}<0$ for some $m$, then $Z$ is feasible but
strictly suboptimal.
\end{proposition}

\begin{IEEEproof}
Since the atom costs are positive and additive, $S\subseteq Z$ implies
$C(S)\leq C(Z)$. Thus, if $B_{i\to j,t}\geq C(Z)$, every subset of $Z$
is feasible. Under monotonicity,
$U_{j,t}(S\mid c,Z)\leq U_{j,t}(Z\mid c,Z)$ for all $S\subseteq Z$,
so $Z$ is optimal.

If $B_{i\to j,t}<C(Z)$, then $Z$ is infeasible. Finally, if
$\varphi^{\LOO}_{i\to j,m,t}<0$, then
$C(Z^{(-m)})<C(Z)\leq B_{i\to j,t}$ and, by~\eqref{eq:loo},
\[
U_{j,t}(Z^{(-m)}\mid c,Z)-U_{j,t}(Z\mid c,Z)
=-\varphi^{\LOO}_{i\to j,m,t}>0.
\]
so $Z^{(-m)}$ is feasible and strictly dominates $Z$.
\end{IEEEproof}
\subsection{Prospective Utility Identification and Surrogate Estimation}
\label{sec:identification}

The utility in~\eqref{eq:payload-utility} is defined interventionally and is
therefore not directly available from ordinary communication logs. To estimate
it, we collect randomized pre-action communication data under different
delivered payloads. These logged data identify the controlled-delivery utility,
which is then approximated by a transmitter-side surrogate and used to estimate
receiver-specific atom values. 

\paragraph{Prospective logging} For focal link $i\to j$, let $X_{i\to j,t}\in\{0,1\}^{M_{i,t}}$ denote a randomized binary delivery mask, where each entry indicates whether the corresponding atom is delivered, and let $S(X_{i\to j,t})$ denote the resulting payload. In prospective (P) logging, the mask is randomized and applied before receiver $j$ acts, so the logged outcome reflects the receiver's response to the actually delivered payload. For comparison, retrospective (R) logging assigns the mask only after the receiver has acted on the full payload, while retrospective occlusion (RO) removes atoms only from the surrogate input. Thus, only P changes the information available to the receiver before its decision.

\paragraph{Utility identification}
To recover the controlled-delivery utility in~\eqref{eq:payload-utility}
from prospective logging data, let $Y_{j,t}(S)$ denote the potential outcome
under forced delivery of payload $S$, and let $\Xi(Z)$ contain the payloads
queried for LOO and Shapley valuation. Identification requires three conditions:
\emph{consistency}, meaning the observed outcome equals $Y_{j,t}(S)$ for the
payload actually delivered; \emph{conditional randomization}, meaning the
delivery mask is independent of the potential outcomes given $(c,Z)$; and
\emph{positivity}, meaning every queried payload has nonzero probability of
being assigned. Under these conditions,
\begin{equation}
\E_{\rm log}
\!\left[
Y_{j,t}\mid S(X_{i\to j,t})=S,c,Z
\right]
=
U_{j,t}(S\mid c,Z).
\label{eq:identification}
\end{equation}
Thus, the conditional mean outcome observed under prospective logging equals
the controlled-delivery utility in~\eqref{eq:payload-utility}.

\paragraph{Surrogate estimation} Because the same $(c,Z,S)$ configuration may not be observed repeatedly, we learn the identified utility using a scalar \emph{shadow-receiver} model $\widehat V_\theta(\widetilde c,Z,x)\in\mathbb{R}$. The surrogate input $\widetilde c$ augments the public context $c$ with
transmitter-available link, stage, and action-support metadata, while $Z$ and $x$ specify the candidate
set and delivered subset; $\beta_{j,t}$ is excluded. Randomized samples collected across links and epochs form \[ \mathcal D = \{(\widetilde c_n,Z_n,x_n,y_n)\}_{n=1}^{N_{\mathcal D}}, \qquad N_{\mathcal D}=|\mathcal D|. \] The surrogate is trained as \begin{equation} \widehat{\theta}\in\operatorname*{arg\,min}_{\theta} \frac{1}{N_{\mathcal D}} \sum_{n=1}^{N_{\mathcal D}} \left[ y_n-\widehat V_\theta(\widetilde c_n,Z_n,x_n) \right]^2 . \label{eq:surrogate-loss} \end{equation} By~\eqref{eq:identification}, this regression estimates $U_{j,t}(S(x)\mid c,Z)$ for payloads supported by the randomized design. For $M=|Z|$, let $\mathbf 1_M$ denote the full delivery mask and $\mathbf 1_M^{(-m)}$ the same mask with atom $m$ withheld. The resulting plug-in LOO score is \begin{equation} v^V_{i\to j,m,t} \triangleq \widehat V_{\widehat\theta}(\widetilde c,Z,\mathbf 1_M) - \widehat V_{\widehat\theta}(\widetilde c,Z,\mathbf 1_M^{(-m)}). \label{eq:plugin-loo} \end{equation} This score estimates $\varphi^{\LOO}_{i\to j,m,t}$; Shapley scores are obtained analogously from surrogate predictions over the required coalition masks. At deployment, transmitter $i$ computes these scores without access to $\beta_{j,t}$, the receiver's realized action or outcome, or an environment rollout.
\subsection{Budgeted Payload Construction}

The estimated atom scores from Sec.~\ref{sec:identification} are used to
construct a budget-feasible payload for each active link $i\to j$.
Atoms with nonpositive scores are excluded, and the remaining atoms are
selected by the following $0$--$1$ knapsack problem:
\begin{equation}
\max_{x\in\{0,1\}^{M}}
\sum_{m=1}^{M} v^V_{i\to j,m,t}x_m
\quad\mathrm{s.t.}\quad
\sum_{m=1}^{M} b_{i,m,t}x_m\leq B_{i\to j,t}.
\label{eq:score-selection}
\end{equation}
With uniform costs, \eqref{eq:score-selection} selects the largest positive
scores up to the affordable cardinality; heterogeneous integer costs admit
exact $0$--$1$ knapsack dynamic programming. LOO requires $M+1$ surrogate
queries per receiver, or $R(M+1)$ for $R$ receivers, while dynamic programming
costs $O(MB)$ in integer cost units. With exact atom values, additive selection is exact under modular utility; with
learned scores or nonmodularity, it approximates~\eqref{prob:selection}. Accordingly, experiments evaluate each selected payload using an
evaluation-only reference estimate of $U_{j,t}$ in~\eqref{eq:payload-utility},
which is unavailable at deployment, rather than the surrogate score objective.

%% file: sections/04_experiments.tex
\section{Experimental Setup}
\label{sec:experiments}

We instantiate Secs.~\ref{sec:model}--\ref{sec:method} in the
Multi-Receiver Intersection Response (MRIR) benchmark, described in detail
below, and evaluate valuation fidelity and budgeted payload selection.
\subsection{MRIR System Instantiation}
\label{subsec:mrir-setup}

MRIR instantiates one Dec-POMDP decision epoch with one roadside transmitter
$i$, four scheduled receivers $j\in\mathcal{R}_{i,t}$, and three
hazards. Each instance contains global state $s_t$, public receiver
contexts $c_{j,t}$, and receiver-private visibility/noise states
$\beta_{j,t}$, which affect receiver decisions but are unavailable to the
transmitter. One error-free pre-action communication round is used.
Each receiver chooses among go, slow, brake, evade-left, and evade-right.
The shared reward is separable,
$r(s_t,\mathbf a_t)=\sum_j r_j(s_t,a_{j,t})$.
Each receiver-local reward $r_j$ takes the value $1$, $0.4$,
$0$, or $-1$ for safe/efficient, over-conservative, inefficient,
or unsafe behavior, respectively, giving $|Y_{j,t}|\leq2$
in~\eqref{eq:receiver-outcome}. MRIR therefore isolates
semantic-layer decisions rather than detailed traffic or channel dynamics.

Semantic extraction is fixed and not trained or evaluated. The mapping in
\eqref{eq:candidate-set} yields a common $Z_{i,t}$ with
$M_{i,t}\in\{8,9,10\}$ atoms from geometry/motion, class/intent,
signal/congestion, duplicate, and stale/conflicting categories, with
receiver-independent confidence and age metadata. Categories describe content,
not utility; negative $\varphi^{\LOO}_{i\to j,m,t}$ arises only when an atom
worsens the frozen receiver policy's decision. All atoms cost $48$ bits.
The receiver policy uses a permutation-invariant set encoder over delivered
atoms, an encoder of the receiver's local information, and a categorical
head over the five actions. It observes $\beta_{j,t}$ and $c_{j,t}$, but
never withheld atoms. It is trained by supervised learning of the
safe/efficient target action computed from the latent scene, on $40{,}000$
instances with randomized partial payloads. Model selection uses only the
policy-validation split; the selected policy is then frozen across all
logging and valuation regimes. Policy-training and validation instances
are disjoint from logging and test instances.
On $8{,}000$ validation instances, the same policy achieves $39.2\%$ accuracy
without the current payload, $58.4\%$ with randomized payloads, and $67.2\%$
with full payload. Table~\ref{tab:setup} summarizes the configuration.
\begin{table}[t]
\caption{MRIR experimental configuration.}
\label{tab:setup}
\centering
\scriptsize
\setlength{\tabcolsep}{2.6pt}
\renewcommand{\arraystretch}{0.90}
\begin{tabular}{@{}ll@{}}
\toprule
\multicolumn{2}{@{}l}{\textbf{Environment and semantic atoms}}\\
Tx / receivers / hazards / actions & $1$ / $4$ / $3$ / $5$\\
Atoms $M_{i,t}$ / categories & $8$--$10$ / $5$\\
Packet cost & $48$ bits/atom\\
Local reward levels / outcome bound &
$\{1,0.4,0,-1\}$ / $|Y_{j,t}|\leq2$\\

\midrule
\multicolumn{2}{@{}l}{\textbf{Fixed receiver policy and logging}}\\
Policy train / val.\ instances & $40{,}000$ / $8{,}000$\\
Policy accuracy: local / masked / full & $39.2$ / $58.4$ / $67.2\%$\\
Logging train / val.\ instances & $25{,}000$ / $6{,}000$\\
Prospective masks: LOO / prefix / Bernoulli &
$35/35/30\%$, $p=0.5$\\

\midrule
\multicolumn{2}{@{}l}{\textbf{Shadow receiver}}\\
Architecture / embedding width & DeepSets / $64$\\
Parameters / independent seeds & $44{,}351$ / $5$\\

\midrule
\multicolumn{2}{@{}l}{\textbf{Evaluation}}\\
Valuation / exact-selection cells & $8{,}000$ / $4{,}000$\\
Stored utilities / oracle draws & $1{,}882{,}112$ / $1{,}408$\\
Budget grid $\alpha$ & $\{0,.125,.25,.375,.5,.75,1\}$\\
Bootstrap / sign tolerance & $10{,}000$ / $\delta_0=0.02$\\
\bottomrule
\end{tabular}
\vspace{-0.2cm}
\end{table}

\subsection{Prospective Logging and Surrogate Training}
\label{subsec:logging-training}

Following Sec.~\ref{sec:identification}, the prospective regime draws
$X_{i\to j,t}$ before receiver $j$ acts using a fixed mixture of $35\%$
LOO-targeted masks for full/one-withheld queries, $35\%$ permutation-prefix
masks for Shapley coalitions, and $30\%$ independent Bernoulli masks with
$p=0.5$ for broad subset support. The mixture is fixed before held-out
evaluation. Logging uses $25{,}000$ training and $6{,}000$ validation
instances, yielding $100{,}000$ and $24{,}000$ focal-link records per regime
and seed across four receivers. Each record is $(\widetilde c,Z,x,Y)$.

\noindent\textbf{Training-label access and computation.}
The offline logger has simulator access to $s_t$, the receiver's local
information, the frozen policy, and the known reward function. After
pre-action delivery of $S(x)$, receiver $j$ draws its action from the
payload-conditioned policy. The logger obtains
$\rho_{j,t}=\pi_j(\cdot\mid\tau^-_{j,t})$ by evaluating the same checkpoint
with an empty current payload and the same local information; no separate
baseline policy is trained. It computes the baseline in
\eqref{eq:receiver-outcome} by exact summation over the five counterfactual
actions, holding $s_t$ and $\mathbf a_{-j,t}$ fixed. Because MRIR's reward
is separable, the nonfocal terms cancel, yielding
\begin{equation}
Y_{j,t}=r_j(s_t,a_{j,t})
-\sum_{a\in\mathcal A_j}\rho_{j,t}(a)r_j(s_t,a).
\label{eq:mrir-label}
\end{equation}
Thus, labels require privileged offline access and are not assumed locally
observable at the roadside transmitter. The training record exposes only
the scalar label to the regressor; neither $s_t$, $\beta_{j,t}$, nor
$\rho_{j,t}$ is a surrogate input. Practical training requires a simulator
or an equivalent reward-evaluation and receiver-feedback mechanism.

The shadow receiver $\widehat V_\theta$ is a permutation-invariant DeepSets
regressor \cite{zaheer2017deep} with width $64$ and $44{,}351$ parameters.
All regimes use the same capacity and training protocol, five independent
seeds, and validation mean-squared error (MSE) for model selection. Deployment uses only transmitter-available $(\widetilde c,Z,x)$;
$\beta_{j,t}$, realized actions/outcomes, and the true utility
$U_{j,t}$ are never inputs.

\subsection{Evaluation Protocol and Result Definitions}
\label{subsec:evaluation-protocol}

We evaluate whether shadow-receiver predictions recover receiver-specific atom
values and whether these values yield effective budgeted payloads. Imperfect
estimates may still support effective selection if they preserve
decision-relevant signs and rankings.

\noindent\textbf{Valuation fidelity.}
We use $2{,}000$ held-out instances across four receivers, giving $8{,}000$
receiver cells per seed, and report five seeds. For each queried payload $S$,
$U^{\mathrm{ref}}_{j,t}(S\mid c,Z)$ approximates~\eqref{eq:payload-utility} by
forcing $S$ and averaging the resulting $Y_{j,t}$ over $\beta_{j,t}$ and
$\omega_t$, using $1{,}408$ conditional draws of $(s_t,\beta_{j,t})$, exact
expectation over five receiver actions, and common draws across payloads. The
reference and estimated LOO effects are
$\varphi^{\LOO,\mathrm{ref}}_{m}=U^{\mathrm{ref}}_{j,t}(Z\mid c,Z)-U^{\mathrm{ref}}_{j,t}(Z^{(-m)}\mid c,Z)$ and
$\widehat\varphi^{\LOO}_{m}=\widehat V_{\hat\theta}(\widetilde c,Z,\mathbf1)-\widehat V_{\hat\theta}(\widetilde c,Z,\mathbf1^{(-m)})$.
We report mean absolute error (MAE),
$N_\varphi^{-1}\sum_{r,m}|\widehat\varphi^{\LOO}_{r,m}-\varphi^{\LOO,\mathrm{ref}}_{r,m}|$, where $N_\varphi$ is the total number of evaluated atom effects,
sign accuracy over $|\varphi^{\LOO,\mathrm{ref}}_{r,m}|>\delta_0$ with
$\delta_0=0.02$, and within-cell Spearman $\rho$. P, R, and RO use matched
capacity, data scale, and seeds.

\noindent\textbf{Budgeted payload evaluation.}
A disjoint set of $1{,}000$ instances across four receivers gives $4{,}000$
cells per seed. For each cell, all $2^{M_{i,t}}$ subsets are evaluated using
the same evaluation-reference utility procedure described above, yielding
$1{,}882{,}112$ stored utilities. We use
$\alpha=B_{i\to j,t}/C(Z_{i,t})
\in\{0,.125,.25,.375,.5,.75,1\}$, with intermediate budgets mapped to
integer cost units. For receiver-specific LOO, the shadow receiver estimates
$\widehat{\varphi}^{\LOO}_{m}
=\widehat V_{\hat\theta}(\widetilde c,Z,\mathbf 1)
-\widehat V_{\hat\theta}(\widetilde c,Z,\mathbf 1^{(-m)})$.
These scores are used in~\eqref{eq:score-selection} to construct the
budget-feasible payload $\widehat S^{\LOO}(B)$. Its performance is then
evaluated independently of the surrogate using
$U^{\mathrm{ref}}_{j,t}(\widehat S^{\LOO}(B)\mid c,Z)$.

All comparators satisfy the same budget. Random Feasible samples uniformly
from feasible subsets; Shared LOO averages the four receiver-specific LOO
vectors; receiver-specific Shapley uses Sec.~\ref{sec:estimand} scores; and
full payload uses $Z_{i,t}$ when feasible. OracleSubset is the evaluation-only
exhaustive reference.
Primary comparisons are paired across seeds, instances, and receivers;
$10{,}000$ hierarchical-bootstrap replicates give $95\%$ percentile
confidence intervals (CIs).

%% file: sections/05_results.tex
\section{Results}
\label{sec:results}

\subsection{Semantic Atom Valuation Fidelity}
\label{subsec:valuation-results}

\begin{figure}[t]
\centering
\includegraphics[width=0.94\columnwidth]{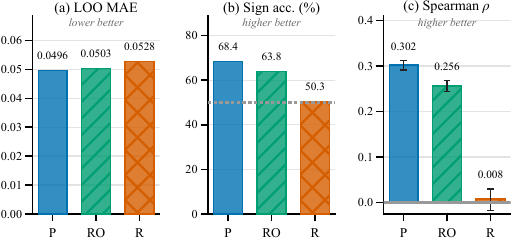}
\caption{LOO fidelity on $8{,}000$ cells; $\delta_0=0.02$; Spearman $\rho$ with $95\%$ CIs.}
\label{fig:valuation}
\vspace{-0.2cm}
\end{figure}

Figure~\ref{fig:valuation} compares estimated and reference LOO effects using
Sec.~\ref{subsec:evaluation-protocol}. In Fig.~\ref{fig:valuation}(a), P, RO,
and R attain MAE $0.0496$, $0.0503$, and $0.0528$. Thus, P differs from the
reference by about $0.05$ utility units per atom on average, while the small
cross-regime differences show that magnitude error alone weakly separates the
logging designs.

Selection-relevant differences are clearer in sign and rank. In
Fig.~\ref{fig:valuation}(b), P recovers $68.4\%$ of non-negligible signs versus
$63.8\%$ for RO and $50.3\%$ for R, gains of $4.6$ and $18.1$ percentage
points. Hence, P distinguishes beneficial from harmful atoms for roughly two
thirds of effects above $\delta_0=0.02$, retaining useful but imperfect
information for~\eqref{eq:score-selection}. Figure~\ref{fig:valuation}(c) gives
$\rho=0.302$ for P with $95\%$ CI $[0.291,0.312]$, versus $0.256$ for RO and
$0.008$ for R, whose interval crosses zero. P therefore preserves a consistent
but moderate ordering, whereas R provides essentially none. Thus, similar MAE can hide selection-relevant sign and rank differences.

This pattern agrees with Sec.~\ref{sec:identification}: P changes the payload
before the receiver acts and targets~\eqref{eq:payload-utility}, whereas RO and
R retain the full-payload decision in retrospective comparisons. P is therefore the strongest regime, while its remaining error motivates the
budgeted test.

\subsection{Budgeted Controlled Utility}
\label{subsec:utility-results}

\begin{figure}[t]
\centering
\includegraphics[width=0.96\columnwidth]{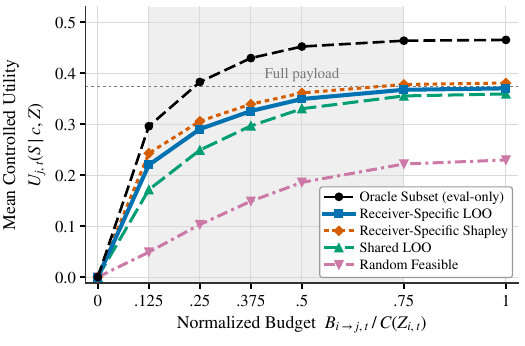}
\caption{Controlled utility vs. $\alpha$; OracleSubset is evaluation-only; Full payload is horizontal.}
\label{fig:budget-utility}
\vspace{-0.2cm}
\end{figure}

Figure~\ref{fig:budget-utility} shows whether the imperfect values in
Sec.~\ref{subsec:valuation-results} yield useful payloads. Utility generally
rises with $\alpha$ as larger budgets admit more useful atoms; OracleSubset is
the exhaustive reference.

Receiver-specific LOO exceeds Shared LOO and Random Feasible across interior
budgets, with the gap most visible when budgets are tight. With few affordable
atoms, utility depends strongly on ranking the most useful ones correctly; as
$\alpha$ grows, ranking errors matter less. The gain over Random Feasible shows
value beyond budget feasibility, while the gain over Shared LOO isolates
receiver-specific ordering because it uses the same atoms, surrogate family,
seeds, budget, and selector but one averaged ranking. 
OracleSubset remains above learned LOO, with the gap consistent with surrogate
error (Fig.~\ref{fig:valuation}) and limitations of scalar-score selection
under nonadditive utility, while the learned gains still confirm useful
downstream ranking.
At $\alpha=1$, full payload yields
$0.3739$ versus $0.4653$ for OracleSubset, about $24\%$ higher. Since all atoms
are affordable, detrimental content can make full transmission suboptimal,
consistent with Proposition~\ref{prop:fullpayload}; this does not imply learned
LOO attains the OracleSubset gain. Receiver-specific Shapley is numerically above LOO across budgets, but no
multiplicity-controlled superiority claim is made.

\subsection{Scope and Outlook}
\label{subsec:limitations}

MRIR evaluates receiver-specific selection in a controlled, single-round
setting with fixed policies. For larger candidate sets, LOO and
permutation-sampled Shapley offer alternatives to exhaustive valuation,
while batching and caching can reduce computational overhead.

\balance
Future work includes heterogeneous packet costs, learned semantic
representations, and channel-aware valuation. Dynamic settings would
require updating the surrogate as policies or context distributions
change. Potential applications include cooperative perception and
robotic coordination.

%% file: sections/06_conclusion.tex
\section{Conclusion}
\label{sec:conclusion}

We studied receiver-specific post-encoding semantic selection under per-link
budgets, identifying controlled-delivery utility by prospective randomization
and estimating signed leave-one-out (LOO) and Shapley values at the transmitter without
receiver-private information.

On Multi-Receiver Intersection Response (MRIR), prospective LOO preserved useful sign and ranking information, and
receiver-specific selection outperformed shared and random alternatives.
Exhaustive subset selection also exceeded full payload at full affordability,
showing that transmitting all available content need not maximize utility.

Useful payload decisions can still arise from imperfect atom values when
selection-relevant signs and rankings are preserved. Future work will
investigate larger receiver populations and learned semantic
representations for cooperative perception and robotic coordination,
permutation-sampled valuation, heterogeneous packet costs, and
channel-aware logging. Adaptive policies and multi-round communication
will require updated utility models beyond the fixed, myopic setting
studied here.

%% file: sections/07_acknowledgment.tex
\section*{Acknowledgment}
This work has been supported by the Federal Ministry of Research,
Technology and Space of the Federal Republic of Germany as part of the
project Open6GHub+ under contract 16KIS2404, and 6G-Sensoria.
The work contributes to the research within the 6G-Valley innovation
cluster.